\documentclass[journal, twocolumn]{IEEEtran}
\usepackage{amsmath,amsthm,amsfonts}
\usepackage{graphicx}
\usepackage{epsf}
\usepackage{mdwmath}
\usepackage{mdwtab}
\usepackage{cite}
\usepackage{amssymb}
\usepackage{algorithmic}
\usepackage{algorithm}
\usepackage{amssymb}
\usepackage{algorithmic}
\usepackage{algorithm}
\usepackage{bm}
\usepackage{color}
\newtheorem{thm}{Theorem}
\newtheorem{cor}{Corollary}

\newcommand{\E}{{\cal{E}}}

\newcommand{\gl}{\gamma_{\mathrm{L}}}
\newcommand{\gu}{\gamma_{\mathrm{U}}}

\def\E{\mathbb{E}}
\def\P{\mathbb{P}}
\def\1{\bm{1}}

\IEEEoverridecommandlockouts

\begin{document}

\title{The Sample Complexity of Search \\ over Multiple Populations}
\author{\IEEEauthorblockN{Matthew L. Malloy, \emph{Member, IEEE}, Gongguo Tang, \emph{Member, IEEE}, Robert D. Nowak, \emph{Fellow, IEEE}
 \thanks{Manuscript received September 06, 2012; revised March 27, 2013; accepted April 08, 2013. This work was partially supported by AFOSR grant FA9550-09-1-0140, NSF grant CCF-1218189 and the DARPA KECoM program.  The material in this paper was presented in part at the Conference on Information Sciences and Systems (CISS) in Princeton, USA \cite{malloyCISS}, March 2012.
\newline \indent M. L. Malloy, G. Tang, and R. D. Nowak are with the Department of Electrical and Computer Engineering, University of Wisconsin, Madison, WI 53715 USA (e-mail: mmalloy@wisc.edu, gtang5@wisc.edu, nowak@engr.wisc.edu).
}}}

\maketitle

\begin{abstract}
This paper studies the sample complexity of searching over multiple populations. We consider a large number of populations, each corresponding to either distribution $P_0$ or $P_1$.   
The goal of the search problem studied here is to find one population corresponding to distribution $P_1$ with as few samples as possible. The main contribution is to quantify the number of samples needed to correctly find one such population.  We consider two general approaches: non-adaptive sampling methods, which sample each population a predetermined number of times until a population following $P_1$ is found, and adaptive sampling methods, which employ sequential sampling schemes for each population.  We first derive a lower bound on the number of samples required by any sampling scheme.  We then consider an adaptive procedure consisting of a series of sequential probability ratio tests, and show it comes within a constant factor of the lower bound.  We give explicit expressions for this constant when samples of the populations follow Gaussian and Bernoulli distributions.   An alternative adaptive scheme is discussed which does not require full knowledge of $P_1$, and comes within a constant factor of the optimal scheme. For comparison, a lower bound on the sampling requirements of any non-adaptive scheme is presented. 


\end{abstract}

\vspace{.2cm}
\noindent 

\section{Introduction}
This paper studies the sample complexity of finding a population corresponding to some distribution $P_1$ among a large number of populations corresponding to either distribution $P_0$ or $P_1$.
More specifically, let $i=1,2,\dots$ index the populations. Samples of each population follow one of two distributions, indicated by a binary label $X_i$: if $X_i = 0$, then samples of population $i$ follow distribution $P_0$, if $X_i = 1$, then samples follow distribution $P_1$.   
We assume that $X_1,X_2,\dots$ are independently and identically distributed (i.i.d.) Bernoulli random variables with $\P(X_i=0)=1-\pi$ and $\P(X_i=1)=\pi$.  Distribution $P_1$ is termed the \emph{atypical} distribution, which corresponds to \emph{atypical} populations, and the probability $\pi$ quantifies the occurrence of such populations.  The goal of the search problem studied here is to find an atypical population with as few samples as possible.

In this search problem, populations are sampled a (deterministic or random) number of times, in sequence, until an atypical population is found.  The total number of samples needed is a function of the sampling strategy, the distributions, the required reliability, and  $\pi$.  To build intuition, consider the following.  As the occurrence of the atypical populations becomes infrequent, (i.e. as $\pi \rightarrow 0$), the number of samples required to find one such population must, of course, increase.   If $P_0$ and $P_1$ are extremely different (e.g., non-overlapping supports), then a search procedure could simply proceed by taking one sample of each population until an atypical population was found.  The procedure would identify an atypical population with, on average,  $\pi^{-1}$ samples.  More generally, when the two distributions are more difficult to distinguish, as is the concern of this paper, we must take multiple samples of some populations.  As the required reliability of the search increases, a procedure must also take more samples to confirm, with increasing certainty, that an atypical population has been found.  

The main contribution of this work is to quantify the number of samples needed to correctly find one atypical population.   {Specifically, we provide matching upper and lower bounds  (to within a constant factor) on the expected number of samples required to find a population corresponding to $P_1$ with a specified level of certainty.  }
We pay additional attention to this sample complexity as $\pi$ becomes small (and the occurrence of the atypical populations becomes \emph{rare}). We consider two general approaches to find an atypical population, both of which sample populations in sequence.  \emph{Non-adaptive} procedures sample each population a predetermined number of times, make a decision, and if the null hypothesis is accepted then move on to the next population.  \emph{Adaptive} methods, in contrast, enjoy the flexibility to sample each population sequentially, and thus, the decision to continue sampling a particular population can be based on prior samples.  

The developments in this paper proceed as follows. First, using techniques from sequential analysis, we derive a lower bound on the expected number of samples needed to reliably identify an atypical population.  To preview the results, the lower bound implies that any procedure (adaptive or non-adaptive) is unreliable if it uses fewer than $\pi^{-1} D(P_0||P_1)^{-1}$ samples on average,  where $D(P_0||P_1)$ is the Kullback-Leibler divergence. We then prove this is tight by showing that a series of sequential probability ratio tests (which we abbreviate as an S-SPRT) succeeds with high probability if the total number of  samples is within a constant factor of the lower bound, provided a minor constraint on the log-likelihood statistic is satisfied (which holds for bounded distributions, Gaussian, exponential, among others).  We give explicit expressions for this constant in the Gaussian and Bernoulli cases.  In the Bernoulli case, the bound derived by instantiating our general results produces the tightest known bound.
In many real world problems, insufficient knowledge of the distributions $P_0$ and $P_1$ makes implementing an S-SPRT impractical. To address this shortcoming, we propose a more practical adaptive procedure known as sequential thresholding, which does not require precise knowledge of $P_1$, and is particularly well suited for problems in which occurrence of an atypical population is rare.  We show sequential thresholding is within a constant factor of optimal in terms of dependence on the problem parameters as $\pi \rightarrow 0$.  Both the S-SPRT procedure and sequential thresholding are shown to be robust to imperfect knowledge of $\pi$. Lastly, we show that non-adaptive procedures require at least $\pi^{-1}  D(P_1||P_0)^{-1} \log\pi^{-1}$ samples to reliably find an atypical population, a factor of $\log\pi^{-1}$ more samples when compared to adaptive methods.

\subsection{Motivating Applications}

Finding an atypical population arises in many relevant problems in science and engineering.  One of the main motivations for our work is the problem of spectrum sensing in cognitive radio.  In cognitive radio applications, one is interested in finding a vacant radio channel among a potentially large number of occupied channels.   Only once a vacant channel is identified can the cognitive device transmit, and thus, identifying a vacant channel as quickly as possible is of great interest.  A number of works have looked at various adaptive methods for spectrum sensing in similar contexts, including \cite{tajer2012adaptive, zhang2010adaptive, laiquickest}.  

Another captivating example is  the \emph{Search for Extraterrestrial Intelligence} (SETI) project.  Researchers at the SETI institute use large antenna arrays to sense for narrowband electromagnetic energy from distant star systems, with the hopes of finding extraterrestrial intelligence with technology similar to ours.  The search space consists of a virtually unlimited number of stars, over 100 billion in the Milky Way alone, each with 9 million potential ``frequencies'' in which to sense for narrow band energy.  The prior probability of extraterrestrial transmission is indeed very small (SETI has yet to make a contact), and thus occurrence of atypical populations is rare.    Roughly speaking, SETI employs a variable sample size search procedure that repeatedly tests energy levels against a threshold up to five times \cite{SETI, NYtimes}.  If any of the measurements are below the threshold, the procedure immediately passes to the next frequency/star.  This procedure is closely related to \emph{sequential thresholding} \cite{MalloyISIT}. Sequential thresholding results in substantial gains over fixed sample size procedures and, unlike the SPRT, it can be implemented without perfect knowledge of $P_1$.

\subsection{Related Work}

The prior work most closely related to the problem investigated here is that by Lai, Poor, Xin, and Georgiadis \cite{5961845}, in which the authors also examine the problem of quickest search across multiple populations, but do not focus on quantifying the sample complexity.  The authors show that the S-SPRT (also termed a CUSUM test) minimizes a linear combination of the expected number of samples and the error probability.  Complementary to this, our contributions include providing tight lower bounds on the expected number of samples required to achieve a desired probability of error, and then showing the sample complexity of the S-SPRT comes within a constant of this bound.   This quantifies how the number of samples required to find an atypical population depends on the distributions $P_0$ and $P_1$ and the probability $\pi$, which was not explicitly investigated in \cite{5961845}.  As a by-product, this proves the optimality of the S-SPRT.


An instance of the quickest search problem was also studied recently in \cite{2012arXiv1202.3639C}, where the authors investigate the problem of finding a biased coin with the fewest flips.   Our more general results are derived using different techniques, and cover this case with $P_0$ and $P_1$ as Bernoulli distributions.  In \cite{2012arXiv1202.3639C}, the authors present a bound on the expected number of flips needed to find a biased coin.  The bound derived from instantiating our more general theory  (see example 2 and Corollary \ref{corr:baisedcoin}) is a minimum of 32 times tighter than the bound in \cite{2012arXiv1202.3639C}.


Also closely related is the problem of sparse signal support recovery from point-wise observations \cite{MalloyISIT, Haupt, AsilomarLimits, malloy2012sequentialJ}, classical work in optimal scanning theory \cite{1057718, 1053868}, and work on pure exploration in multi-armed bandit problems \cite{bubeck2009pure, mannor2004sample}.  The sparse signal recovery problems differ in that the total number of populations is finite, and the objective is to recover all (or most) populations following $P_1$, as opposed to finding a single population and terminating the procedure.  Traditional multi-armed bandit problems differ in that no knowledge of the distributions of the arms is assumed.

\section{Problem Setup}
Consider an infinite number of populations indexed by $i =1,2,...$. For each population $i$, samples of that population are distributed either 
\begin{eqnarray}  \nonumber
Y_{i,j} &\overset{iid}{\sim}& P_0 \; \; \text{if} \; \: X_i =0 
\quad \text{or} \quad \\ \nonumber
Y_{i,j} &\overset{iid}{\sim}& P_1  \; \; \text{if} \; \: X_i =1
\end{eqnarray}
where $P_0$ and $P_1$  are probability measures supported on $\mathcal{Y}$, $j$ indexes multiple i.i.d. samples of a particular population, and $X_i$ is a binary label.   The goal is to find a population $i$ such that $X_i = 1$ as quickly and reliably as possible.  The prior probability of a particular population $i$ following $P_1$ or $P_0$ is i.i.d., and denoted
\begin{eqnarray} \nonumber
\mathbb{P}\left( X_i =1 \right) &=& \pi \\
\mathbb{P}\left( X_i=0 \right) &=& 1-\pi \nonumber
\end{eqnarray}
where we assume $\pi \leq 1/2$ without loss of generality.

A testing procedure samples a subset of the populations and returns a single index, denoted $I$.  The performance of any testing procedure is characterized by two metrics: \emph{1)} the expected number of samples required for the procedure to terminate, denoted $\mathbb{E}[N]$, and \emph{2)} the probability of error, defined as
\begin{eqnarray} \nonumber
P_e := \mathbb{P}\left( I \in \left\{i: X_i =0 \right\} \right). 
\end{eqnarray}  
 In words, $P_e$ is the probability a procedure returns an index that does not correspond to a population following $P_1$.

\begin{figure}
\vspace*{3mm}\hrule \vspace*{1mm}
Search for an atypical population
\vspace*{1mm}
\hrule \vspace*{1mm}
\begin{algorithmic}[] \label{alg:sds}
\STATE{initialize: $i = 1$, $j=1$}
\WHILE{atypical population not found}
\STATE{\bf{either}}
\STATE{\hspace{.2cm} {\bf{1) sample $Y_{i,j}$}}, set $j = j+1$}
\STATE{\hspace{.2cm} {\bf{2) move to next population}}: $i = i+1$, $j=1$}
\STATE{\hspace{.2cm} {\bf{3) terminate}}: $\hat{X}_i = 1$}
\ENDWHILE
\STATE{output: $I = i$}
\end{algorithmic}
\vspace*{1mm}\hrule \vspace*{1mm}
\caption{Search for an atypical population.  The search procedures under consideration are restricted to the above framework.
\label{fig:frame} }
\end{figure}

In order to simplify analysis, we make two assumptions on the form of the procedures under consideration. 

\vspace{.2cm}
\noindent {\bf{Assumption 1.}} Search procedures follow the framework of Fig. \ref{fig:frame}. Specifically, a procedure starts at population $i=1$.  Until termination, a procedure then $\emph{1)}$ takes a sample of $i$, or $\emph{2)}$ moves to index $i+1$, or \emph{3)}, terminates, declaring population $i$ as following distribution $P_1$ (deciding $\hat{X}_i = 1$).  

\vspace{.2cm} \noindent
Assumption 1 implies procedures do not revisit populations. It can be argued that this restricted setting has no loss of optimality when $P_1$, $P_0$, and $\pi$ are known; in this Bayesian setting, the posterior probability of population $i$ depends only on samples of that index.  This posterior reflects the probability of error if the procedure were to terminate and estimate $\hat{X}_i = 1$.  Since this probability is not affected by samples of other indices, for any procedure that \emph{returns} to re-measure a population, there is a procedure requiring fewer samples with the same probability of error that either did \emph{not} return to index $i$, or did not move away from index $i$ in the first place.   Note that \cite{5961845} makes the same assumption.

The second assumption we make is on the invariance of the procedure across indices.  To be more specific, imagine that a procedure is currently sampling index $i$. For a given sampling procedure,  if $X_i = 1$, the probability the procedure passes to index $i+1$ without terminating is denoted $\beta$, and the probability the procedure correctly declares $\hat{X}_i=1$ is $1-\beta$.   Likewise, for any $i$ such that $X_i =0$, the procedure falsely declares $\hat{X}_i = 1$ with probability $\alpha$, and continues to index $i+1$ with probability $1 -\alpha$.  

\vspace{.2cm}
\noindent {\bf{Assumption 2.} }
$\alpha$ and $\beta$ are invariant as the procedure progresses; i.e., they are not functions of the index under consideration.

\vspace{.2cm}
\noindent Under Assumption 2, provided the procedure arrives at population $i$, we can write
\begin{eqnarray} \nonumber
\beta &=& \mathbb{P}(\hat{X}_i = 0 | X_i = 1) \\ \nonumber
\alpha &=&\mathbb{P}(\hat{X}_i = 1 | X_i = 0).
\end{eqnarray}
Note that this restriction has no loss of optimality as the known optimal procedure \cite{5961845} has this form.  Restricted to the above framework, a procedure consists of a number of simple binary hypothesis tests, each with false positive probability $\alpha$ and false negative probability $\beta$.  While any pair $(\alpha, \beta)$ and $\E[N]$ do parameterize the procedure, our goal is to develop universal bounds in terms of the underlying problem parameters, $P_e$ and $\E[N]$.   

Assumptions 1 and 2 allow for the following recursive relationships, which will be central to our performance analysis.  Let $N_i$ be the (random) number of samples taken of population $i$, and $N = \sum_{i=1}^\infty N_i$ be the total number of samples taken by the procedure.  We can write the expected number of samples as
\begin{eqnarray} \label{eqn:en1232}
\mathbb{E}[N] &=&   \E[N_1]  \quad + \\ \nonumber 
 & &  \hspace{-.6cm}\E\left[N_2+N_3 + ... \: \left \vert \hat{X}_1 = 0 \right.\right] \left( (1-\pi) (1-\alpha) +\pi \beta \right)
\end{eqnarray}
where $ (1-\pi) (1-\alpha) +\pi \beta $ is the probability the procedure arrives at the second index.
The expected number of samples used from the second index onwards, given that the procedure arrives at the second index (without declaring $I = 1$), is simply equal to the total number of samples:  $\E[N_2+N_3 + ... \: | \hat{X}_1 = 0 ] = \E[N]$.  Rearranging terms in (\ref{eqn:en1232}) gives the following relationship
\begin{eqnarray} \label{eqn:expNumMeas}
\E[N] &=& \frac{ \E[N_1]   }{ \alpha(1-\pi) + \pi (1-\beta)  }.
\end{eqnarray}
In the same manner we arrive at the following expression for the probability of error:
\begin{eqnarray} 
P_e = \frac{ \alpha(1-\pi)}{ \alpha(1-\pi) + \pi (1-\beta) } = \frac{1}{1+\frac{\pi(1-\beta)}{\alpha(1-\pi)}} . \label{eqn:Pe} 
\end{eqnarray}
From this expression we see that if  
\begin{eqnarray} \nonumber
\frac{\alpha(1-\pi)}{\pi(1-\beta)} \geq \delta
\end{eqnarray}
for some $\delta > 0$, then $P_e \geq \frac{\delta}{1+\delta}$,
and $P_e$ is greater than or equal to some positive constant.

  
Lastly, the bounds derived throughout often depend on explicit constants, in particular the Kullback-Leibler divergence between the distributions, defined in the usual manner:
\begin{eqnarray} \nonumber
D(P_1||P_0) &:=& \E_1 \left[ L(Y) \right]  \\
D(P_0||P_1) &:=& \E_0 \left[ - L(Y) \right]  \nonumber
\end{eqnarray}   
where 
\begin{eqnarray} \nonumber
L(Y) := \log \frac{P_1(Y)}{P_0(Y)}
\end{eqnarray}
is the log-likelihood ratio. 
Other constants are denoted by $C_1$, $C_1'$, etc., and represent distinct numerical constants which may depend on $P_0$ and $P_1$.

\section{Lower bound for any procedure}
We begin with a lower bound on the number of samples required by any procedure to find a population following distribution $P_1$.   The main theorem of the section is presented first, followed by two corollaries aimed at highlighting the relationship between the problem parameters. 


\begin{thm} \label{thm:SeqLB}
Any procedure with
\begin{eqnarray} \nonumber
 P_e \leq \frac{\delta}{1+\delta}
\end{eqnarray}
also has 
\begin{eqnarray}  \label{eqn:thmMain}
\E[N] &\geq &   \frac{1-\pi}{ \: \pi  } \frac{(1-\delta)^2}{(1+\delta)}   \max \left(1, \frac{1}{D(P_0||P_1)} \right)   + \\ \nonumber
&& \frac{\log\left( \frac{1}{2\pi \delta} \right)}{ D(P_1||P_0)} \left( \frac{1-\delta\frac{D(P_1||P_0)}{D(P_0||P_1)}  }{1+\delta} \right)    -  \frac{1}{D(P_1||P_0)}
\end{eqnarray}
for any $\delta \in [0,1/2]$.
\end{thm}
\begin{proof}
See Appendix A.
\end{proof}   
Theorem \ref{thm:SeqLB} lower bounds the expected number of samples required by any procedure to achieve a desired performance, and is comprised of two terms with dependence on $\pi$ and the error probability, and a constant offset.   To help emphasize this dependence on the problem parameters, we present the following Corollary.
\begin{cor} \label{corr:LB}
Any procedure with
\begin{eqnarray} \nonumber
 P_e \leq \frac{\delta}{1+\delta}
\end{eqnarray}
also has
\begin{eqnarray} \label{eqn:corrLB} 
\E[N]  &\geq&  \hspace{-.2cm} \frac{1}{D(P_0||P_1) }\left(\frac{1}{12 \: \pi } + \frac{1}{3} \; \log \left( \frac{1}{2 \pi \delta}\right)   -  1\right)
\end{eqnarray}
for any $\delta \leq 1/2$.  Here, we assume $D(P_0||P_1) = D(P_1||P_0)$ for simplicity of presentation. 
\end{cor}

\noindent Proof of Corollary \ref{corr:LB} follows immediately from Theorem \ref{thm:SeqLB}, as $\pi \leq 1/2$ and $\delta \leq 1/2$.  

Corollary \ref{corr:LB} provides a particularly intuitive way to quantify the number of samples required for the quickest search problem. The first term in (\ref{eqn:corrLB}), which has a $1/\pi$ dependence, can be interpreted as the minimum number of samples required to \emph{find} a population following distribution $P_1$.  The second term, which has a $\log \delta^{-1}$ dependence, is best interpreted as the minimum number of samples required to \emph{confirm} that a population following $P_1$ has been found. 

When the populations following distribution $P_1$ become rare (when $\pi$ tends to zero),  the second and third terms in (\ref{eqn:corrLB})  become small compared to the first term.  This suggests the vast majority of samples are used to \emph{find} a rare population, and a vanishing proportion are needed for \emph{confirmation}.  The corollary below captures this effect.  The leading constants are of particular importance, as we relate them to upper bounds in Sec. \ref{sec:SSPRT}. 
In the following, consider $P_e$ and $\E[N]$ as functions $\pi$, $P_0$, $P_1$, and some sampling procedure $\mathcal{A}$.

\begin{cor}  {\bf{Rare population.}}   \label{corr:LB2}
Fix $\delta \in (0,1/2]$. Then any procedure $\mathcal{A}$ that satisfies
\begin{eqnarray} \nonumber
\limsup_{\pi \rightarrow 0}  P_e  \leq \frac{\delta}{1+\delta}
\end{eqnarray}  \label{lem:LBrare}
also has 
\begin{eqnarray} \nonumber
\liminf_{\pi \rightarrow 0}  \; \pi \: \E[N] & \geq &  \frac{(1-\delta)^2}{(1+\delta)}  \max \left(\frac{1}{D(P_0||P_1)}, 1   \right).
\end{eqnarray}
\end{cor}
\noindent The proof of Corollary \ref{lem:LBrare} follows from Theorem \ref{thm:SeqLB} by noting both the second and third terms of (\ref{eqn:thmMain}) are overwhelmed as $\pi$ becomes small.   

Corollary \ref{lem:LBrare} is best interpreted in two regimes: \emph{(1)} the high SNR regime, when $D(P_0||P_1) > 1$, and \emph{(2)}, the low SNR regime, when $D(P_0||P_1) \leq 1$. 
In the high SNR regime, when $D(P_0||P_1) > 1$, any procedure with $\lim_{\pi \rightarrow 0} P_e = 0$
also has $\lim_{\pi \rightarrow 0} \: \pi \: \mathbb{E}[N] \geq 1$.
This simply implies that any procedure requiring fewer samples in expectation than $\pi^{-1}$ also has probability of error bound away from zero.   The bound becomes tight when the SNR becomes high -- when $D(P_0||P_1)$ is sufficiently large, we expect to classify each population with one sample.    
In the lower SNR regime, where $D(P_0||P_1) \leq 1$, any procedure with $\lim_{\pi \rightarrow 0} P_e = 0$ also has 
$\lim_{\pi \rightarrow 0} \: \pi \: \mathbb{E}[N] \geq  1/D(P_0||P_1)$.
In the low SNR regime the sampling requirements are at best an additional factor of $D(P_0||P_1)^{-1}$ higher than when we can classify each distribution with one sample.
   

\section{S-SPRT Procedure} \label{sec:SSPRT}
The Sequential Probability Ratio Test (SPRT), optimal for simple binary hypothesis tests in terms of minimizing the expected number of samples for tests of a given power and significance \cite{Wald1948}, can be applied to the problem studied here by implementing a series of SPRTs on the individual populations.  For notational convenience, we refer to this procedure as the S-SPRT.  This is equivalent in form to the CUSUM test studied in \cite{5961845},  which is traditionally applied to change point detection problems. 

The S-SPRT operates as follows.  Imagine the procedure has currently taken $j$ samples of population $i$. The procedure continues to sample population $i$ provided 
\begin{eqnarray} \label{eqn:SPRTrule}
  \gamma_\mathrm{L} \ < \ \Lambda_{i,j} \ < \ \gamma_\mathrm{U}.
\end{eqnarray}
where $\Lambda_{i,j} :=  \prod_{k=1}^{j} \frac{P_1(Y_{i,k})}{P_0(Y_{i,k})} $ is the likelihood ratio statistic, and $\gu$ and $\gl$ are scalar upper and lower thresholds.
In words, the procedure continues to sample population $i$ provided the likelihood ratio comprised of samples of that population is between two scalar thresholds.   The S-SPRT stops sampling population $i$ after $N_i$ samples, which is a random integer representing the smallest number of samples such that (\ref{eqn:SPRTrule}) no longer holds:
\begin{eqnarray} \nonumber
N_i : = \min\left\{j: \, \Lambda_{i,j} \leq \gl \; \bigcup \; \Lambda_{i,j}\geq \gu \right\}.
\end{eqnarray}
When the likelihood ratio exceeds (or equals) $\gu$, then $\hat{X}_i = 1$, and the S-SPRT terminates returning $I=i$.  Conversely, if the likelihood ratio falls below (or equals) $\gl$, then $\hat{X}_i = 0$, and the procedure moves to index $i+1$.  The procedure is detailed in Algorithm \ref{alg:ssprt}.

\begin{algorithm}[h]
\caption{\hspace{1cm}Series of SPRTs Procedure (S-SPRT)}
\begin{algorithmic} \label{alg:ssprt}
\STATE{input: thresholds $\gl$, $\gu$, distributions $P_0$, $P_1$}
\STATE{initialize: $i = 1$, $j=1$, $\Lambda = 1$ }
\WHILE{$ \Lambda \ < \ \gamma_\mathrm{U}$}
\STATE{{\bf measure}: $Y_{i,j}$}
\STATE{{\bf{compute}}: $\Lambda = \Lambda \cdot \frac{P_1(Y_{i,j} )}{ P_0(Y_{i,j} )} $}
\IF{$\Lambda \leq \gamma_{L}  $ }
\STATE{$i = i+1$, $j = 1$, $\Lambda =1$}
\ELSE
\STATE{$j = j+1$}
\ENDIF
\ENDWHILE
\STATE{output: $I=i$}
\end{algorithmic}
\end{algorithm}

The S-SPRT procedure studied in \cite{5961845} fixes the lower threshold in each individual SPRT at $\gl = 1$ (and hence terms the procedure a CUSUM test). This has a very intuitive interpretation; since there are an infinite number of populations, anytime a sample suggests that a particular population does not follow $P_1$, moving to another population is best.   While this approach is optimal \cite{5961845}, we use a strictly smaller threshold, as it results in a simpler derivation of the upper bound.  

In the following theorem and corollary we assume a minor restriction on the tail distribution of the log-likelihood ratio test statistic, a notion studied in depth in \cite{SPRTbounds1960}.  Specifically, recall $L = \log (P_1(Y) /P_0(Y))$ is the log-likelihood statistic.  We require that
\begin{eqnarray} \label{eqn:ass1}
\max_{r\geq 0} \; \E\left[L - r |L \geq r\right] < \infty
\end{eqnarray}
and 
\begin{eqnarray} \label{eqn:ass2}
\min_{r\geq 0} \; \E\left[L + r |L \leq -r \right] > -\infty.
\end{eqnarray}
This condition is satisfied when $L$ follows any bounded distribution, Gaussian distributions, exponential distributions, among others.  It is not satisfied by distributions with infinite variance or polynomial tails.  A more thorough discussion of this restriction is studied in \cite{SPRTbounds1960}.

\begin{thm} \label{thm:SPRT}
The S-SPRT with $\gl \in (0,1)$ and 
$\gu = \frac{1-\pi}{\pi \delta}$, $\delta \in [0,1/2]$ satisfies  
\begin{eqnarray} \nonumber
P_e \leq \frac{\delta}{1+\delta}
\end{eqnarray}
and 
\begin{eqnarray} \label{eqn:SSPRTbnd}
 \E[N]  \leq \frac{ C_1  }{\pi }  +  \frac{ \log \frac{1}{\pi \delta} }{D(P_1||P_0)}  + C_2
\end{eqnarray} 
for some constants $C_1$ and $C_2$ independent of $\pi$ and $\delta$.
\end{thm}
\begin{proof}
The full proof is given in Appendix B.   
\end{proof}
The main argument of the proof of Theorem \ref{thm:SPRT}  follows from standard techniques in sequential analysis.  The constant $C_1$ is a function of the underlying distributions and is given by
\begin{eqnarray} \label{eqn:CC2}
C_1 = \frac{C_1' + \log \gl^{-1}}{( 1-\gl)D(P_0||P_1) } 
\end{eqnarray}
where $C_1'$ is a bound on the \emph{overshoot} in the log-likelihood ratio when it falls outside $\gamma_U$ or $\gamma_L$.  $C_1'$, and thus $C_1$, can be explicitly calculated depending on the underlying distributions in a number of cases (see Examples 1 and 2, and \cite[page 145]{1945Wald}, and \cite{SPRTbounds1960}).

\begin{cor} {\bf{Rare population.}} \label{lem:SPRTS}
Fix $\delta \in (0,1/2]$. The S-SPRT with any $\gl \in (0,1)$ and  
$\gu = \frac{1-\pi}{\pi \delta}$  satisfies 
$ P_e \leq \frac{\delta}{1+\delta} $
and 
\begin{eqnarray} \nonumber
 \lim_{\pi \rightarrow 0} \; \pi \:\E[N]  \leq  C_1 
\end{eqnarray}
for some constant $C_1$ independent of $\pi$ and $\delta$.
\end{cor}
The proof of Corollary \ref{lem:SPRTS} is an immediate consequence of Theorem \ref{thm:SPRT}.  Note that $\gu>1$, since we assume $\pi \leq 1/2$.  As the atypical populations become rare, sampling is dominated by \emph{finding} an atypical population, which is order $\pi^{-1}$.   The constant factor of $C_1$ is the multiplicative increase in the number of samples required when the problem becomes noisy.

\vspace{.3cm}
\noindent{\bf{Remark 1.}} The S-SPRT procedure is fairly insensitive to our knowledge of the true prior probability $\pi$. On one hand, if we overestimate $\pi$ by using a  larger $\tilde{\pi}$ to specify the upper threshold $\gu = \frac{1-\tilde{\pi}}{\tilde{\pi}\delta}$, then according to \eqref{eqn:pe} the probability of error $P_e$ increases and is approximately $\frac{\tilde{\pi}}{\pi}\frac{\delta}{1+\delta}$, while the order of $\mathbb{E}[N]$ remains the same. On the other hand, if our $\tilde{\pi}$ underestimates $\pi$, then the probability of error $P_e$ is reduced by a factor of $\tilde{\pi}/\pi$, and the order of $\mathbb{E}[N]$ also remains the same, provided  $\log(1/\tilde{\pi}) \leq 1/\pi$, i.e., $\tilde{\pi}$ is not exponentially smaller than $\pi$. As a consequence, it is sensible to underestimate $\pi$, rather than overestimate $\pi$ as the latter would increase the probability of error.

\vspace{.3cm}
\noindent \textbf{Remark 2.}  Implementing a sequential probability ratio test on each population can be challenging for many practical problems.  While the S-SPRT is optimal when both $P_0$ and $P_1$ are known and testing a single population amounts to a simple binary hypothesis test, scenarios often arise where some parameter of distribution $P_1$ is unknown.  Since the SPRT is based on exact knowledge of $P_1$, it cannot be implemented in this case.  A simple example where $P_0 \sim \mathcal{N}(0,1)$ and $P_1 \sim \mathcal{N}(\mu,1)$, for some unknown $\mu > 0$, illustrates this issue.    
 Many alternatives to the SPRT have been proposed for \emph{composite} hypothesis (see \cite{lai1988nearly, fellouris2012almost}, etc.).  In the next section we propose an alternative that is near optimal and also simple to implement.

\section{Sequential Thresholding}

Sequential thresholding, first proposed for sparse recovery problems in \cite{MalloyISIT}, can be applied to the search for an atypical population, and admits a number of appealing properties.  It is particularly well suited for problems in which the atypical distributions are rare.  Sequential thresholding does not require full knowledge of the distributions, specifically $P_1$, as required by the S-SPRT (see Remarks 2 and 4).  Moreover, the procedure admits a general error analysis, and perhaps most importantly is very simple to implement (a similar procedure is used in the SETI project \cite{SETI, NYtimes}).  The procedure can substantially outperform non-adaptive procedures as $\pi$ becomes small. Roughly speaking, for small values of $\pi$, the procedure reliably recovers an atypical population with
\begin{eqnarray} \nonumber
\E[N] \lesssim \frac{C }{\pi}
\end{eqnarray}
for some constant $C$ independent of $\pi$.

\begin{algorithm}[h]
\caption{\hspace{1cm} Sequential Thresholding}
\begin{algorithmic} \label{alg:sds}
\STATE{input: integer $k_{\max}$}
\STATE{initialize: $i = 1$, $k = 1$}
\WHILE{$k \leq k_{\max}$}
\STATE{{\bf measure}: $(Y_{i,j+1},\dots,Y_{i,j+k})$ where $j = \sum_{m=1}^k (m-1)$ }
\IF{$T \left( Y_{i,j+1},\dots,Y_{i,j+k}\right) \leq \gamma_k$ }
\STATE{$i = i+1$}
\STATE{ $k = 1$}
\ELSE
\STATE{$k = k+1$}
\ENDIF
\ENDWHILE
\STATE{output: $\hat{X}_i = 1$}
\end{algorithmic}
\end{algorithm}

Sequential thresholding requires one input: $k_{\max}$, an integer representing the maximum number of \emph{rounds} for any particular index.  
Let $T$ represent a sufficient statistic for the likelihood ratio that does not depend on the parameters of $P_1$ or $P_0$ (for example, when $P_0$ and $P_1$ are Gaussian with different mean, $T(Y_{i,j},...,Y_{i,j'}) = Y_{i,j}+\dots+Y_{i,j'}$).

The procedure searches for an atypical population as follows.  Starting on population $i$, the procedure takes one sample.  If the sufficient statistic comprised of that sample is greater than the threshold, i.e. $T(Y_{i,1}) > \gamma_1$, the procedures takes two additional samples of index $i$ and forms $T(Y_{i,2},Y_{i,3})$ (which is only a function of the second and third samples). If $T(Y_{i,2},Y_{i,3}) > \gamma_2$, three more samples are taken, and $T(Y_{i,4},Y_{i,5},Y_{i,6})$ is tested against a threshold.  The procedure continues in this manner, taking $k$ samples on \emph{round} $k$, and testing the statistic \emph{up to} a maximum of $k_{\max}$ times.  If the statistic is below the threshold, i.e. $T_k < \gamma_k$, on any round, the procedure immediately moves to the next population, setting $i = i+1$, and resetting $k$.  Should any population survive all $k_{\max}$ rounds, the procedure estimates $\hat{X}_i = 1$, and terminates.  The procedure is detailed in Algorithm \ref{alg:sds}.

Control of the probability of error depends on the series of thresholds $\gamma_k$ and the number of rounds $k_{\max}$. For our analysis the thresholds are set as to satisfy    
\begin{eqnarray} \nonumber
\P_0\left(T_k > \gamma_k  \right) =\frac{1}{2}.
\end{eqnarray}
In practice, the thresholds can be set in any way such that test statistic under $P_0$ falls below the threshold with fixed non-zero probability.

Intuitively, the procedure controls the probability of error as follows.  First, $\alpha$ can be made small by increasing $k_{\max}$; as each round is independent, $\alpha =(1/2)^{k_{\max}}$.  Of course, as $k_{\max}$ is increased, $\beta$ also increases.   Fortunately, as $k_{\max}$ grows, it can be shown that $\beta$ is strictly less than one  (provided the Kullback-Leibler divergence between the distributions is non-zero).   The following theorem quantifies the number of samples required to guarantee recovery of an index following $P_1$ as $\pi$ grows small.

\begin{thm}{\bf{Sequential Thresholding}}. \label{thm:SeqThres}
Sequential thresholding with ${k_{\max}} =\left  \lceil 2 \log_{2} \left(  \frac{1-\pi}{  \pi}\right) \right \rceil$ satisfies 
\begin{eqnarray} \nonumber
\lim_{\pi \rightarrow 0} P_e  = 0
\end{eqnarray}
 and 
 \begin{eqnarray} 
\nonumber
\lim_{\pi \rightarrow 0} \pi \;\E[N] \leq C
\end{eqnarray}
for some constant $C$ independent of $\pi$.
\end{thm}
\begin{proof} 
See Appendix C.
\end{proof}

 \ \\ \noindent{\bf{Remark 3.}} Similar to the behavior of the SPRT discussed in Remark 1, sequential thresholding is also fairly insensitive to our prior knowledge of $\pi$, especially when we underestimate $\pi$. More specifically, overestimating $\pi$ increases the probability of error almost proportionally and has nearly no affect on $\mathbb{E}[N]$, while underestimating $\pi$ decreases the probability of error and the order of $\mathbb{E}[N]$ is the same as long as $\log(1/\tilde{\pi}) \leq 1/\pi$.

\vspace{.3cm}
\noindent{\bf{Remark 4.}}  For many distributions in the exponential family, the log-likelihood ratio, $L$, is a monotonic function of a test statistic $T$ that does not depend on parameters of  $P_1$. As a consequence of the sufficiency of $T$, the thresholds $\gamma_k$, $k = 1,\dots,k_{\max}$, depend only on $P_0$, making sequential thresholding suitable when knowledge about $P_1$ is not available.   

Perhaps most notably, in contrast to the SPRT based procedure, sequential thresholding does not aggregate statistics.  Roughly speaking, this results in increased robustness to modeling errors in $P_1$ at the cost of a sub-optimal procedure. Analysis of sequential thresholding in related sparse recovery problems can be found in \cite{MalloyISIT, AsilomarLimits}.

\section{Limitations of Non-Adaptive Procedures}
For our purposes a non-adaptive procedure tests each individual population with a pre-determined number of samples, denoted $N_0$.  In this case, the conditional number of samples for each individual test is simply $ \E[N_1|X_1 = 0] = \E[N_1|X_1 = 1] = N_0$, giving
\begin{eqnarray} \label{eqn:ENfixed}
\E[N] = \frac{N_0}{ \alpha (1- \pi) +\pi(1  - \beta) }.
\end{eqnarray}
To compare the sampling requirements of non-adaptive procedures to adaptive procedures, we present a necessary condition for reliable recovery.
The theorem implies that non-adaptive procedures require a factor of $\log \pi^{-1}$ more samples than the best adaptive procedures.  

\begin{thm}{\bf{Non-adaptive procedures}}. \label{thm:NSLB}
 Any non-adaptive procedure that satisfies
\begin{eqnarray} \nonumber
P_e \leq \frac{\delta}{1+\delta}
\end{eqnarray}
also has
\begin{eqnarray} \nonumber
\E[N] \geq  \frac{ \log\left(\frac{1}{2 \delta \pi } \right) -1 }{ \pi (1+\delta) D(P_1||P_0)} .
\end{eqnarray}
for $\delta \leq 1/2$.
\end{thm}

\begin{proof}
See Appendix D.
\end{proof}

\vspace{.3cm}
\noindent{\bf{Remark 5.}}
The lower bound  presented in Theorem \ref{thm:NSLB} implies that non-adaptive procedures require at best a multiplicative factor of $\log \pi^{-1}$ more samples than adaptive procedures (as adaptive procedures are able to come within a constant factor of the lower bound in Theorem \ref{thm:SeqLB}).  For problems with even modestly small values of $\pi$, this can result in non-adaptive sampling requirements many times larger than those required by adaptive sampling procedures.

\section{Examples and Numerical Results}

\vspace{.3cm}
\noindent {\bf{Example 1.}} \emph{Searching for a Gaussian with positive mean}. Consider searching for a population following $P_1 \sim \mathcal{N}(\mu,1)$ amongst a number of populations following $P_0 \sim \mathcal{N}(-\mu,1)$ for some $\mu > 0$. The Kullback-Leibler divergence between the Gaussian distributions is $D(P_0||P_1) = 2\mu^2$.  

Focusing on the S-SPRT (Alg. \ref{alg:ssprt}) and Theorem \ref{thm:SPRT},  we have an explicit expression for $C_1'$ (defined in (\ref{eqn:CC2})) based on the overshoot of the likelihood ratio \cite[page 145]{1945Wald}:
\begin{eqnarray} \label{eqn:example1C1}
C_1'(\mu) = 2\mu\left(\mu + \frac{e^{-\mu^2/2}}{\int_{-\mu}^\infty e^{-t^2/2}dt}\right).
\end{eqnarray}
In order to make our bound on $\mathbb{E}[N]$ as tight as possible, we would like to minimize $C_1$ from (\ref{eqn:CC2}) with respect to $\gl$.  Since the minimizer has no closed form expression, we use the sub-optimal value $\gl = 1/\mu$ for $\mu > 1$, and $\gl = 1-\sqrt{\mu}$ for $\mu <1$. For this choice of $\gl$, the constant $C_1 = C_1(\mu)$ in Theorem \ref{thm:SPRT} and (\ref{eqn:CC2}) is 
\begin{eqnarray} \nonumber
C_1(\mu)  = 
\begin{cases} \frac{C_1'(\mu) + \log(\mu)}{ (1-1/\mu)D(P_0||P_1)} &\mbox{if }  \mu > 1\\
\frac{C_1'(\mu) + \log((1-\sqrt{\mu})^{-1})}{\sqrt{\mu} D(P_0||P_1)}& \mbox{if } \mu <  1. 
\end{cases}  
\end{eqnarray} 
Consider the following two limits. First, as $\mu \rightarrow \infty$ 
\begin{eqnarray} \nonumber
\lim_{\mu \rightarrow \infty} C_1(\mu)=1.
\end{eqnarray}
As a consequence (from Corollary \ref{lem:SPRTS}) 
\begin{eqnarray} \nonumber
\lim_{\mu\rightarrow \infty} \lim_{\pi\rightarrow 0} \pi \: \mathbb{E}[N] \leq \lim_{\mu \rightarrow \infty} C_1(\mu) = 1.
\end{eqnarray}
This implies Corollary \ref{lem:LBrare} is tight in this regime.  As $\mu$ tends to infinity we approach the noise-free case, and the procedure is able to make perfect decisions with one sample per population.  As expected, the required number of samples grows as $1/\pi$.

Second, as $\mu \rightarrow 0$, 
\begin{eqnarray} \nonumber
\lim_{\mu \rightarrow 0} C_1(\mu) D(P_0||P_1) =1
\end{eqnarray}
which implies (again from Corollary \ref{lem:SPRTS}) 
\begin{eqnarray} \nonumber
\lim_{\mu\rightarrow 0} \lim_{\pi\rightarrow 0} \: \pi \: D(P_0||P_1) \: \mathbb{E}[N] \leq \lim_{\mu \rightarrow 0}  C_1(\mu) D(P_0||P_1) = 1.
\end{eqnarray}
Comparison to Corollary \ref{lem:LBrare} shows the bound is tight.  
For small $\pi$, the S-SPRT requires $1/(\pi D(P_0||P_1))$ samples as the distributions grow similar; no procedure can do better.

Fig. \ref{fig:Norm} plots the expected number of samples scaled by $\pi$ as a function of $\mu$.  Specifically, the figure displays four plots.  First, $\mu$ vs. $\pi \: \mathbb{E}[N]$ obtained from simulation of the S-SPRT procedure is plotted:  $\pi = 10^{-3}$, $\gl = 1$, $\gu = \frac{1 - \pi }{\pi \delta }$ and $\delta = 10^{-2}$.  
Second, the lower bound from Theorem \ref{thm:SeqLB} is shown.  For small $\pi$, from (\ref{eqn:thmMain}), any reliable procedure has
\begin{eqnarray} \nonumber
\mathbb{E}[N]  \gtrsim \frac{1}{\pi} \max \left(1, \frac{1}{D(P_0||P_1)}\right).
\end{eqnarray}   
The upper bound from Theorem \ref{thm:SPRT} is also plotted.  From (\ref{eqn:SSPRTbnd}), for small values of $\pi$, the S-SPRT achieves 
\begin{eqnarray} \nonumber
 \mathbb{E}[N]  \lesssim \frac{C_1}{\pi}.
\end{eqnarray}
where $C_1$ is calculated by minimizing (\ref{eqn:CC2}) over $\gamma_L \in (0,1)$ for each value of $\mu$.  $C_1$ is within a small factor of the lower bound for all values of $\mu$.  

Lastly, the performance of sequential thresholding (Alg. \ref{alg:sds}) is plotted.  The maximum number of round is specified as in Theorem \ref{thm:SeqThres}.

\begin{figure}[htb]
\centerline{
\includegraphics[width=9.7cm]{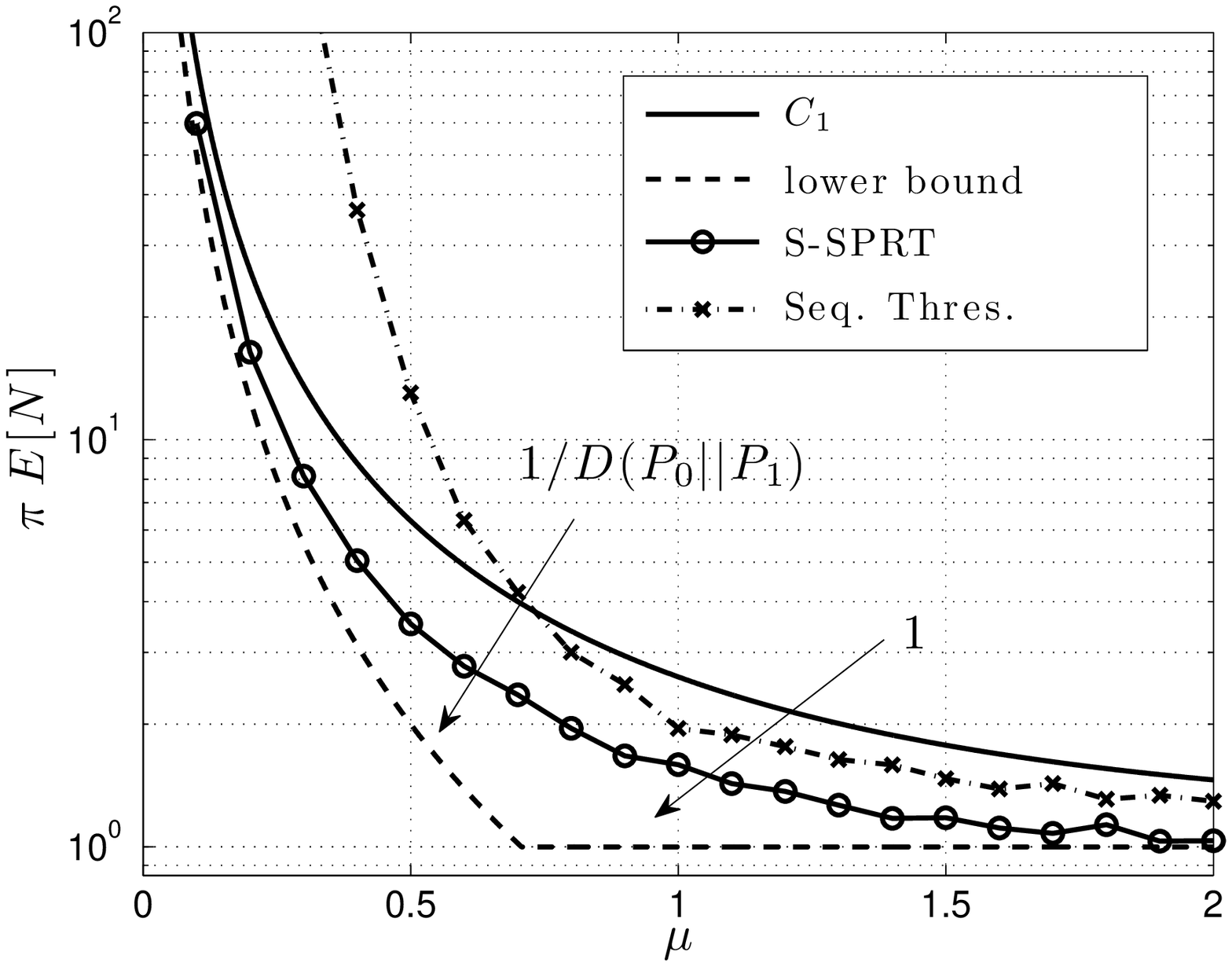}}
\caption{ \label{fig:Norm} Expected number of samples scaled by $\pi$ as a function of the mean of the atypical population, $\mu$, corresponding to example 1.  Simulation of the S-SPRT (Alg. \ref{alg:ssprt}) is plotted with the upper bound from Corollary \ref{lem:SPRTS} and lower bound from Corollary \ref{lem:LBrare}.  Sequential thresholding (Alg. \ref{alg:sds}). $\pi = 10^{-3}$,  $P_e \leq 10^{-2}$, $10^3$ trials for each value of $\mu$. }
\end{figure}


\vspace{.3cm}
\noindent {\bf{Example 2.}} \emph{Searching for a biased coin.}  Consider the problem of searching for a coin with bias towards heads of $1/2 + b$ amongst coins with bias towards heads of $1/2 - b$, for $b \in [0, 1/2]$. This problem was studied recently in \cite{2012arXiv1202.3639C}.

\begin{cor} {\bf{Biased Coin}}. \label{corr:baisedcoin}
The S-SPRT procedure (Alg. \ref{alg:sds})  with $\gl = \frac{1-2b}{1+2b}$  and $\gu = \frac{1-\pi}{\pi \delta}$ satisfies $P_e \leq \frac{\delta}{1+\delta}$ and 
\begin{eqnarray} \nonumber
\mathbb{E}[N] \leq \frac{1}{2b^2} \left(\frac{1}{\pi}   +  \log \left(\frac{1}{\pi \delta} \right) + 1\right).
\end{eqnarray}
\end{cor}
\begin{proof}
The proof follows from evaluation of the constants in  Theorem \ref{thm:SPRT}.
The log-likelihood ratio corresponding to each sample (each coin flip) takes one of two values: if a coin reveals heads, $L = \log \frac{1 + 2 b}{1-2 b}$, and if a coin reveals tails, $L = \log \frac{1 - 2 b}{1+2 b}$.  When each individual SPRT terminates, it can exceed the threshold by no more than this value, giving,
\begin{eqnarray}  \nonumber
C_1'(b) = \log \frac{1 + 2 b}{1-2 b}  \qquad C_2'(b) = \log \frac{1 + 2 b}{1-2 b}
\end{eqnarray}
where $C_1'$ is defined in (\ref{eqn:CC2}), and $C'_2$ in (\ref{eqn:CC22}).
With $\gl = \frac{1 - 2 b}{1+2 b}$, we can directly calculate the constants in Theorem \ref{thm:SPRT}. From (\ref{eqn:CC2}),
\begin{eqnarray} \nonumber
{C_1(b)}= \frac{1+2b}{4b^2} \leq \frac{1}{2b^2}
\end{eqnarray}
as the Kullback-Leibler divergence is $D(P_0||P_1) = D(P_1||P_0) = 2b \log\frac{1+2b}{1-2b}$.
Also note $1/D(P_1||P_0) \leq 1/(2b^2)$.
Lastly, from (\ref{eqn:CC22}), 
\begin{eqnarray} \nonumber
C_2(b)= \frac{C_2'}{D(P_1||P_0)}  = \frac{1}{2b} \leq \frac{1}{2b^2}.
\end{eqnarray} 
Combining these with Theorem \ref{thm:SPRT} completes the proof.
\end{proof}

Comparison of Corollary \ref{corr:baisedcoin}  to  \cite[Theorem 2]{2012arXiv1202.3639C} shows the leading constant is a factor of 32 smaller in the bound presented here.

Moreover, closer inspection reveals that the constant $C_1(b)$ can be further tightened.  Specifically, note that when an individual SPRT estimates $\hat{X}_i = 0$ it must hit the lower threshold exactly (since $\gamma_L = {(1-2b)}/{(1+2b)}$). If we choose only values of $\delta$ such that the upper threshold is an integer multiple of the likelihood ratio (i.e., set $\log \gu =  k\log((1+2b)/(1-2b)) $ for some integer $k$) the overshoot here is also zero.  $C_1' = 0$ and $C_2' =0$, which then give
\begin{eqnarray} \label{eqn:C1bb}
{C_1(b)}=  \frac{1+2b}{8b^2}.
\end{eqnarray}
From Corollary \ref{lem:SPRTS},  
\begin{eqnarray}  \label{eqn:coinUB}
 \lim_{\pi \rightarrow 0} \pi \:\E[N]  \leq \frac{1+2b}{8b^2}.
\end{eqnarray}
For small $\pi$, the number of samples required by any procedure to reliably identify an atypical population is 
\begin{eqnarray} \nonumber
\E[N] \lesssim \frac{1}{\pi} \left(\frac{1+2b}{8b^2}\right).
\end{eqnarray}
If $b = 1/2$ (each coin flip is deterministic), $C_1 = 1$, and the expected number of samples grows as $1/ \pi$ as expected.   The upper bound in Corollary \ref{lem:SPRTS} and lower bound in Corollary \ref{corr:LB2} converge.

Likewise, as the bias of the coin becomes small, $\lim_{b\rightarrow 0} C_1(b) D(P_0||P_1) = 1$, and the expected number of samples needed to reliably identify an atypical population grows as $1/( \pi D(P_0||P_1))$.  Again the upper and lower bounds converge. 

Note that the S-SPRT procedure for testing the coins in this particular example is equivalent to a simple, intuitive procedure, which can be implemented as follows: beginning with coin $i$, and a scalar static $T = 0$, if heads appears, add $1$ to the statistic.  Likewise, if tails appears, subtract $1$ from the test statistic.  Continue to flip the coin until either \emph{1)} $T$ falls below 0, or \emph{2)} $T$ exceeds some upper threshold  (which controls the error rate).  If the statistic falls below 0, move to a new coin, and reset the count, i.e., set $T=0$; conversely if the statistic exceeds the upper threshold, terminate the procedure.  Note that any time the coin shows tails on the first flip, the procedure immediately moves to a new coin.  

\begin{figure}[htb]
\centerline{
\includegraphics[width=9.7cm]{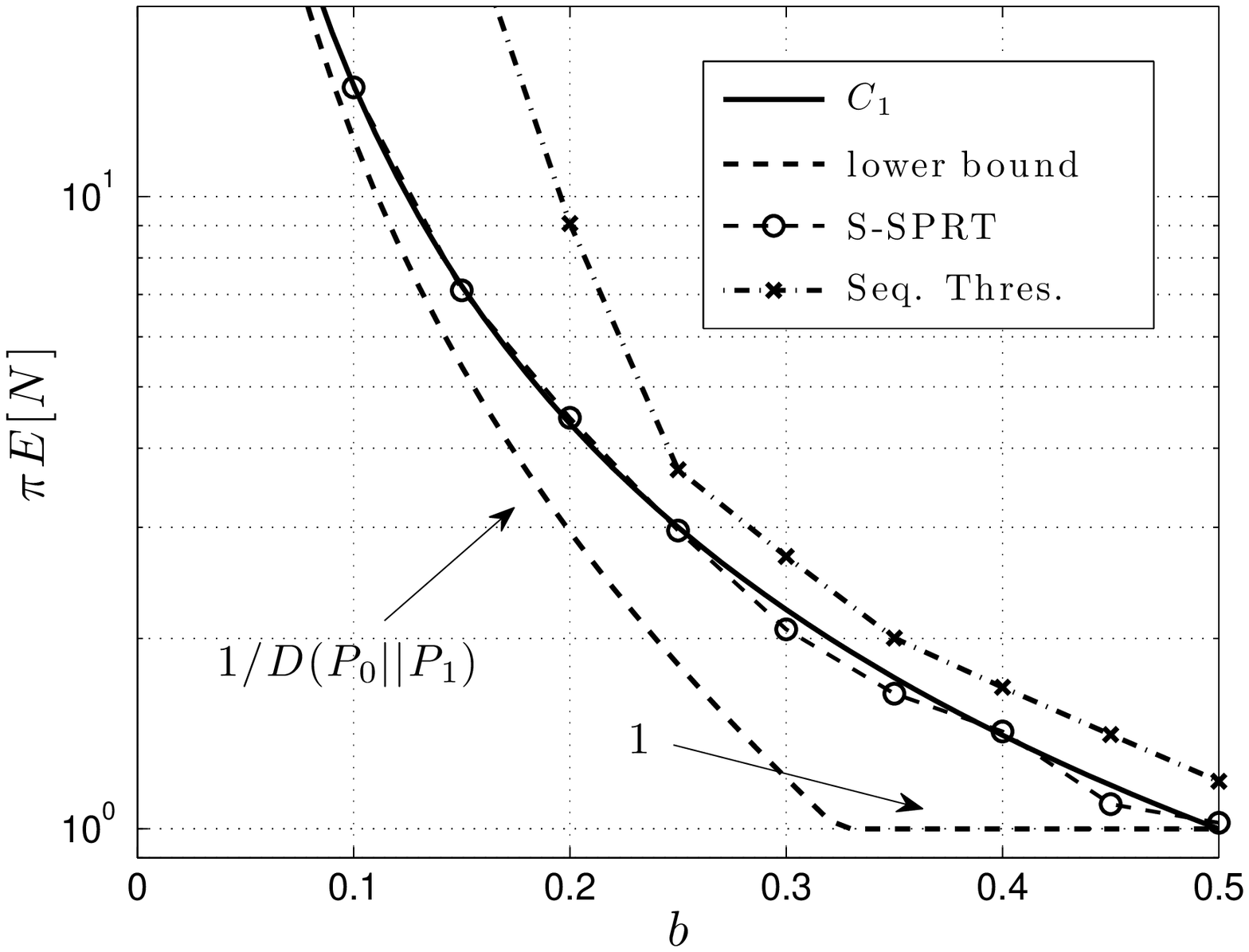}}
\caption{ \label{fig:Coin} Expected number of samples scaled by $\pi$ as a function of the bias of the coin corresponding to Example 2. Upper and lower bounds from  Corollaries \ref{lem:SPRTS} and \ref{lem:LBrare}.   Simulation of the S-SPRT (Alg. \ref{alg:ssprt}).   Simulation of sequential thresholding (Alg. \ref{alg:sds}).  $\pi = 10^{-3}$,  $P_e \leq 10^{-2}$, $10^3$ trials for each value of $b$.}
\end{figure}

Fig. \ref{fig:Coin} plots the expected number of samples scaled by $\pi$ as a function of the bias of the atypical coins, $b$.  The S-SPRT was simulated with the lower threshold set at $\gl =\log \frac{1 - 2 b}{1+2 b}$ for $\pi = 10^{-3}$ and $P_e \leq 10^{-2}$.  The upper and lower bounds from Corollaries \ref{lem:SPRTS} and \ref{corr:LB2} are also plotted.  The upper bound, $C_1$, is given by the expression in (\ref{eqn:C1bb}).

Notice that the simulated S-SPRT procedure appears to achieve the upper bound.  Closer inspection of the derivation of Theorem \ref{thm:SPRT} with $C'_1 = 0$ (as the \emph{overshoot} in (\ref{eqn:E0bnd11}) is zero), shows the bound on the number of samples required by the S-SPRT is indeed tight for the search for the biased coin.  Performance of sequential thresholding (Alg. \ref{alg:sds}) is included for comparison.

\section{Conclusion}
This paper explored the problem of finding an atypical population amongst a number of typical populations, a problem arising in many aspects of science and engineering.  

More specifically, this paper quantified the number of samples required to recover an atypical population with high probability.  We paid particular attention to problems in which the atypical populations themselves become increasingly rare.  After establishing a lower bound based on the Kullback Leibler divergence between the underlying distributions, the number of samples required by the optimal S-SPRT procedure was studied; the number of samples is within a constant factor of the lower bound, which can be explicitly derived in a number of cases.   Two common examples, where the distributions are Gaussian and Bernoulli, were studied.

Sequential thresholding, a more robust procedure that can often be implemented with less prior knowledge about the distributions was presented and analyzed in the context of the quickest search problem.   Sequential thresholding requires a constant factor more samples than the S-SPRT.  Both sequential thresholding and the SPRT procedure were shown to be fairly robust to modeling errors in the prior probability.   Lastly, for comparison, a lower bound for non-adaptive procedures was presented.

\section*{Appendix A}
\emph{Proof of Theorem \ref{thm:SeqLB}:}
Assume that $ P_e \leq \frac{\delta}{1+\delta}$ and from (\ref{eqn:Pe}) we have
\begin{eqnarray} \label{eqn:Seqalphimp}
\frac{\alpha(1-\pi)}{\pi(1-\beta)} \leq \delta.  
\end{eqnarray}
For ease of notation, define
\begin{eqnarray} \label{eqn:EOE1def}
E_{1} = \E[N_1|X_1 = 1]  \qquad E_{0} = \E[N_1|X_1 = 0].
\end{eqnarray}
 From (\ref{eqn:expNumMeas}),
\begin{eqnarray}  \nonumber
\mathbb{E}[N]  &=& \frac{\pi E_1 + (1-\pi) E_0 }{ \alpha (1- \pi) +\pi(1  - \beta) }  \geq  \frac{\pi E_1 + (1-\pi) E_0} {(1+\delta) \pi(1-\beta)   }   \\
& = &\frac{ E_1}{ (1+\delta)(1-\beta) } + \frac{ (1-\pi) E_0}{ (1+\delta) \pi (1-\beta)}. 
 \label{eqn:Ebnd3}
\end{eqnarray} 
From standard sequential analysis techniques (see \cite[Theorem 2.29]{SeqAnalysis}) we have the following identities relating the expected number of measurements to $\alpha$ and $\beta$, which hold for \emph{any} binary hypothesis testing procedure:
\begin{eqnarray} \label{eqn:E1bnd}
E_1 &\geq & \frac{\beta \log \left( \frac{\beta}{1-\alpha} \right)   + (1-\beta) \log\left( \frac{ 1- \beta}{\alpha}\right)  } { D(P_1||P_0)} \\ \label{eqn:E0bnd}
E_0 &\geq & \frac{\alpha \log \left( \frac{\alpha}{1-\beta} \right)   + (1-\alpha) \log\left( \frac{ 1- \alpha}{\beta}\right)  } { D(P_0||P_1)}  .
\end{eqnarray}

Rearranging (\ref{eqn:Ebnd3}),
\begin{eqnarray} \nonumber
\E[N] &\geq&  \underbrace{\frac{\beta \log \left( \frac{\beta}{1-\alpha} \right) }{(1+\delta)   (1-\beta)D(P_1||P_0)} }_{T_1}  +  \underbrace{ \frac{  \log\left( \frac{ 1- \beta}{\alpha}\right)}{ (1+\delta)D(P_1||P_0)}}_{T_2}  \\ \nonumber
&+& \underbrace{\frac{(1-\pi)\left(\alpha \log \left( \frac{\alpha}{1-\beta} \right)   + (1-\alpha) \log\left( \frac{ 1- \alpha}{\beta}\right) \right) } { \pi(1+\delta)(1-\beta)D(P_0||P_1)} }_{T_3} .
\end{eqnarray}
We first bound $T_1$ as
\begin{eqnarray} \label{eqn:T1bnd7}
T_1 \geq \frac{-1}{(1+\delta)D(P_1||P_0)} \geq \frac{-1}{D(P_1||P_0)}
\end{eqnarray}
since for all $\beta\in [0,1]$,
\begin{eqnarray} \nonumber
\frac{\beta \log \frac{\beta}{1-\alpha}}{1-\beta} \geq \frac{\beta \log {\beta}}{1-\beta}  \geq -1.
\end{eqnarray}
From (\ref{eqn:Seqalphimp}),
\begin{eqnarray} \nonumber
T_2 \geq \frac{  \log\left( \frac{ 1- \pi}{\pi \delta}\right)}{ (1+\delta)D(P_1||P_0)}.
\end{eqnarray}
Next, differentiating $T_3$ with respect to $\alpha$ gives 
\begin{eqnarray} \nonumber
\frac{d(T_3)}{d\alpha} = \frac{(1-\pi) \log \frac{\alpha \beta}{(1-\alpha)(1-\beta ) } }{ (1-\delta)\pi (1-\beta)D(P_0||P_1)}
\end{eqnarray}
 showing that the expression is non-increasing in $\alpha$ over the set of $\alpha$ satisfying
$\frac{\alpha}{1-\beta} \leq \frac{1-\alpha}{\beta}.$
From (\ref{eqn:Seqalphimp}), we are restricted to 
$\frac{\alpha}{1-\beta} \leq  \frac{\delta \pi  }{ 1 - \pi} $
and thus, if
$\frac{\delta \pi  }{ 1 - \pi} \leq \frac{1-\alpha}{\beta}$,
then (\ref{eqn:E0bnd}) is non-increasing in $\alpha$.  To show this, note that 
\begin{eqnarray} \nonumber
\frac {\delta \pi}{1-\pi} \leq \delta \leq 1-\delta \leq 1- \frac{\alpha(1-\pi)}{\pi(1-\beta)} \leq 1- \alpha \leq \frac{1-\alpha}{\beta}
\end{eqnarray} 
since both $\delta \leq 1/2$ and $\pi \leq 1/2$.
We can replace $\alpha$ in (\ref{eqn:E0bnd}) with $ \frac{\delta \pi (1-\beta)}{1 - \pi}$.  This gives 
\begin{eqnarray} \nonumber
T_3  &\geq &  \frac{ \delta \log \left(   \frac{\delta \pi}{1 - \pi} \right) }{(1+\delta)   D(P_0||P_1)  } + \\ &&  \frac{(1-\pi) \left( 1-\frac{\delta \pi(1-\beta)}{ 1-\pi} \right) \log\left( \frac{1}{\beta} -  \frac{\delta \pi (1-\beta)}{\beta(1 - \pi)} \right) }{ \pi D(P_0||P_1) (1+\delta) (1-\beta)   } \nonumber \\ \nonumber
 &\geq &  \frac{ \delta \log \left(   \frac{\delta \pi}{1 - \pi} \right) }{(1+\delta)   D(P_0||P_1)  } + \\ &&  \frac{(1-\pi) \left( 1-\delta \right) \log\left( \frac{1}{\beta} -    \frac{\delta (1-\beta) } {\beta } \right) }{ \pi D(P_0||P_1) (1+\delta) (1-\beta)   } \nonumber \\ 
&\geq& \frac{ \delta \log \left(   \frac{\delta \pi}{1 - \pi} \right) }{(1+\delta)   D(P_0||P_1)  }  +  \frac{ (1-\pi)\left( 1-\delta \right)^2  }{ \pi D(P_0||P_1) (1+\delta)  } \nonumber
\end{eqnarray}
where the first inequality follows from making the substitution for $\alpha$ and from (\ref{eqn:Seqalphimp}), and the second inequality follows since $\pi/(1-\pi) \leq 1$ and $1-\beta \leq 1$, and the last inequality follows as
\begin{eqnarray}\label{eqn:elebd}
\frac{\log\left(\frac{1}{\beta} - \frac{\delta(1-\beta)}{\beta}\right)}{(1-\beta)} \geq 1-\delta
\end{eqnarray}
 for all $\beta \in [0,1]$. To see the validity of \eqref{eqn:elebd}, we note
 \begin{eqnarray}
\frac{\log\left(\frac{1}{\beta} - \frac{\delta(1-\beta)}{\beta}\right)}{(1-\beta)(1-\delta)}& = &\frac{\log\left(1 + \frac{(1-\delta)(1-\beta)}{\beta}\right)}{(1-\beta)(1-\delta)}\nonumber\\
& \substack{(\star)\\ \geq} & \frac{\log\left(1+\frac{1-\beta}{\beta}\right)}{1-\beta}\nonumber\\
& = & \frac{\log(1/\beta)}{1-\beta}\nonumber\\
& \geq & 1. \nonumber
 \end{eqnarray}
Here ($\star$) follows by noting that $\log\left(1+x/\beta\right) / x$ is monotonically decreasing in $x$, and by setting $x = 1-\beta$.

We can also trivially bound $E_0$ by noting that $E_0 \geq 1$.  This provides an additionally bound on $T_3$:
\begin{eqnarray} \nonumber
T_3 &\geq& \frac{(1-\pi)}{\pi (1+\delta)(1-\beta)} \\
& \geq &\frac{ \delta \log \left(   \frac{\delta \pi}{1 - \pi} \right) }{(1+\delta)   D(P_0||P_1)  }  +  \frac{ (1-\pi)\left( 1-\delta \right)^2  }{ \pi (1+\delta)  } \label{eqn:bnd1samp}
\end{eqnarray}  
since the first term in (\ref{eqn:bnd1samp}) is strictly negative.

Combining the bounds on $T_1$ and $T_2$, and the two bounds on $T_3$, and noting that $\delta \pi /(1-\pi) \leq 2\delta \pi $  gives 
\begin{eqnarray} \nonumber
\E[N] & \geq &   \frac{1-\pi}{ \: \pi  } \frac{(1-\delta)^2}{(1+\delta)}   \max \left(1, \frac{1}{D(P_0||P_1)} \right)   + \\ \nonumber
&& \frac{\log\left( \frac{1}{2\pi \delta} \right)}{ D(P_1||P_0)} \left( \frac{1-\delta\frac{D(P_1||P_0)}{D(P_0||P_1)}  }{1+\delta} \right)    -  \frac{1}{D(P_1||P_0)}
\end{eqnarray}
completing the proof.

\section*{Appendix B}
\emph{Proof of Theorem \ref{thm:SPRT}:}
The proof is based on techniques used for analysis of the SPRT.  
From \cite{SeqAnalysis}, the false positive and false negative events are related to the thresholds as:
\begin{eqnarray} \label{eqn:SPRTalpha}
\alpha \leq \gu^{-1}(1- \beta) \leq \gu^{-1} = \frac{\pi \delta}{1-\pi}
\end{eqnarray}
\begin{eqnarray} \label{eqn:SPRTbeta}
\beta \leq \gl(1-\alpha) \leq \gl.
\end{eqnarray}
From (\ref{eqn:Pe}) the probability the procedure terminates in error returning a population following $P_0$ is
\begin{eqnarray} 
P_e & \leq &  \frac{1}{1+\frac{\pi}{\gu^{-1}(1-\pi)}} 
= \frac{\delta}{1+\delta}.\label{eqn:pe}
\end{eqnarray}  
To show the second part of the theorem, first define the log-likelihood ratio as
\begin{eqnarray} \label{eqn:LogLR}
L_i^{(j)} = \sum_{k=1}^j \log \frac{P_1(Y_{i,k})}{ P_0(Y_{i,k})}.
\end{eqnarray}
For ease of notation, let $E_0$ and $E_1$ be defined in (\ref{eqn:EOE1def}).
By Wald's identity \cite{SeqAnalysis},
\begin{eqnarray} \nonumber
E_0 & = & \frac{- \E_0 \left[  L_i^{(N_i)}  \right] }{D(P_0||P_1)   }  \quad = \\ \nonumber
& &  \hspace{-.8cm} \frac{ (1-\alpha) \E_0 \left[ \left. -L_i^{(N_i)} \right \vert \hat{X} = 0 \right] + \alpha  \E_0 \left[ \left.- L_i^{(N_i)} \right \vert \hat{X} = 1 \right] }{D(P_0||P_1)}.
\end{eqnarray}
The expected value of the log-likelihood ratio after $N_i$ samples (i.e, when the procedure stops sampling index $i$) is often approximated by the stopping boundaries themselves (see \cite{SeqAnalysis}).  In our case, it is sufficient to show the value of the likelihood ratio when the procedure terminates or moves to the next index can be bound by a constant independent of $\pi$ and $\delta$.  From \cite[Eqns. 4.9 and 4.10]{Wald1948}, for $C_1' \geq 0$,
\begin{eqnarray} \label{eqn:E0bnd11}
\E_0 \left[ \left. L_i^{(N_i)} \right \vert \hat{X} = 0 \right] \geq \log \gl - C_1' 
\end{eqnarray}
and 
\begin{eqnarray} \label{eqn:E1bnd11}
\E_0 \left[ \left. L_i^{(N_i)} \right \vert \hat{X} = 1 \right] \leq \log \gu + C_1' 
\end{eqnarray}
where $C_1'$ is any constant that satisfies both
\begin{eqnarray} \nonumber
C_1' \leq \max_{r \geq 0 } \; \E_0\left[L^{(1)} - r \left \vert L^{(1)} \geq r \right.  \right]
\end{eqnarray}
and
\begin{eqnarray} \nonumber
C_1' \leq \max_{r \geq 0 } \; \E_0\left[- (L^{(1)} + r) \left \vert L^{(1)} \leq -r  \right.  \right].
\end{eqnarray}
$C_1'$ depends only on the distribution of $L^{(1)}$, and is trivially independent of  $\gl$ and $\gu$.  Under the assumptions of (\ref{eqn:ass1}) and (\ref{eqn:ass2}), the constants are finite.  $C_1'$ can be explicitly calculated for a variety of problems (see Examples 1 and 2, and \cite{1945Wald, SPRTbounds1960}).  $C_1'$ is a bound on the \emph{overshoot} in the log-likelihood ratio when it falls outside $\gamma_U$ or $\gamma_L$.  We have 
\begin{eqnarray} \nonumber
E_0 &\leq&  \frac{ (1-\alpha) \left( C_1' + \log \gl^{-1}  \right)  + \alpha  \left( -C_1' + \log  \gu^{-1} \right)   }{D(P_0||P_1)} \\ \nonumber
&\leq &  \frac{ (1-\alpha)(C_1'  + \log \gl^{-1} )  }{D(P_0||P_1)}
\end{eqnarray}
where the second inequality follows as $\gu \geq1$.  Likewise,
\begin{eqnarray} \nonumber
E_1 &\leq&  \frac{ (1-\beta) ( C_2' + \log \gu  )  + \beta  ( -C_2' + \log  \gl )   }{D(P_1||P_0)} \\ \nonumber
&\leq &   \frac{  (1-\beta) (C_2'   + \log \gu  ) }{D(P_1||P_0)}
\end{eqnarray}
for some constant $C_2' \geq 0$ which represents the overshoot of the log-likelihood ratio given $X_i = 0$.
Combining these with (\ref{eqn:expNumMeas}) bounds the expected number of samples:
\begin{eqnarray} \label{eqn:EN55} \nonumber
\E[N] &=& \frac{\pi E_1 + (1-\pi) E_0 }{ \alpha (1- \pi) +\pi(1  - \beta) }   \\ \nonumber
 &\leq & \frac{    \pi  \frac{(1-\beta)( C_2' + \log \gu ) }{  D(P_1||P_0)} + (1-\pi)\frac{ (1-\alpha)( C_1'  + \log \gl^{-1}   ) } { D(P_0||P_1)} }{ \alpha(1-\pi) +  \pi(1-\beta) } \\ \nonumber
& \leq &  \frac{ C_2'+ \log \left(    \frac{1-\pi}{\pi \delta} \right)    }{ D(P_1||P_0)} +\frac{ C_1' + \log \gl^{-1} } {\pi  ( 1- \gl) D(P_0||P_1)} \\
&\leq & \frac{C_1}{\pi } + \frac{ \log \frac{1}{\pi \delta}}{D(P_1||P_0)} + {C_2}\nonumber
\end{eqnarray}
where the second inequality follows from dropping $\alpha(1-\pi)$ from the denominator, replacing $\beta$ with the bound in (\ref{eqn:SPRTbeta}), and dropping $(1-\alpha)(1-\pi)$ from the numerator of the second term.   The third inequality follows from defining
\begin{eqnarray}
 C_2 = \frac{C_2'}{D(P_1||P_0)} \label{eqn:CC22}
 \end{eqnarray}
 and 
\begin{eqnarray} \label{eqn:CC}
C_1 = \frac{C_1' + \log \gl^{-1}}{ (1-\gl)D(P_0||P_1)} 
\end{eqnarray}
completing the proof.

\section*{Appendix C}
\emph{Proof of Theorem \ref{thm:SeqThres}:} Employing sequential thresholding, the false positive event depends on the number of rounds as $\alpha = (1/2)^{k_{\max}}$.  With $k_{\max}$ as specified, we have $\alpha \leq \pi^2 /(1-\pi)^2$ and from (\ref{eqn:Pe}),
\begin{eqnarray} \nonumber
\P_e \leq \frac{\frac{\pi}{1-\pi}}{\frac{\pi}{1-\pi} + 1-\beta }.
\end{eqnarray}
Next, we show that $1-\beta$ is bound away from zero as $\pi$ becomes small.  Since $\lim_{\pi \rightarrow 0} k_{\max} = \infty$, 
\begin{eqnarray} \nonumber
\lim_{\pi \rightarrow 0} 1- \beta 
& =&  \lim_{k_{\max} \rightarrow \infty} \P_1\left( \bigcap_{k=1}^{k_{\max}} T_k \geq \gamma_k \right) \\ \nonumber
& = & \prod_{k=1}^{q} \P_1\left( T_k \geq \gamma_k \right)  \prod_{k=q+1}^\infty \P_1\left( T_k \geq \gamma_k \right)  \\
& > &  0. \label{eqn:STzero}
\end{eqnarray}
The last inequality can be seen as follows.  Since $\P_0\left(T_k \geq \gamma_k\right)$ is fixed by definition, we can apply Stein's Lemma \cite{Cover:1991:EIT:129837}.  For any $\epsilon > 0$, there exists an integer $q$ such that for all $k > q$,
\begin{eqnarray} \nonumber
\P_1\left( T_k \leq \gamma_k \right) \leq e^{-k (1-\epsilon) D(P_0||P_1) }.
\end{eqnarray}
This implies that for sufficiently large $q$,
\begin{eqnarray} \nonumber
 && \hspace{-1cm} \prod_{k=q+1}^\infty  \P_1\left( T_k \geq \gamma_k \right) \\ \nonumber
 & \geq &  \prod_{k=q+1}^\infty 1 - e^{-k (1-\epsilon) D(P_0||P_1) } \\ \nonumber
 & = & \exp\left(\sum_{k=q}^{\infty} \log\left( 1 - e^{-k (1-\epsilon) D(P_0||P_1)} \right) \right) \\ \nonumber 
 &{\geq} & \exp{(-C')}.
\end{eqnarray}
for some constant $C'>0$.  The last inequality follows as the inner sum is convergent to a finite negative value for any $(1-\epsilon)D(P_0||P_1) > 0$.  By assumption of non-zero Kullback-Leibler divergence,  $\P_1\left( T_k \geq \gamma_k \right)  \geq 1/2$ for all $k$.  For any fixed $q$, $\prod_{k=1}^{q} \P_1\left( T_k \geq \gamma_k \right) \geq 1/2^q >0$, and (\ref{eqn:STzero}) holds.  
Given the prescribed $k_{\max}$, we have
\begin{eqnarray}
\lim_{\pi \rightarrow 0} P_e = \lim_{\pi \rightarrow 0} \frac{\frac{\pi}{1-\pi}}{\frac{\pi}{1-\pi} + (1-\beta) } = 0.
\end{eqnarray}
The expected number of samples required for any index following $P_0$, defined in (\ref{eqn:EOE1def}), is given as
\begin{eqnarray} \nonumber
E_0 = \sum_{k=1}^{k_{\max}}   \frac{k}{2^{k-1}}   \leq 4.
\end{eqnarray}
On the other hand, the expected number of samples given the index follows $P_1$ is bound as:
\begin{eqnarray} \nonumber
E_1 \leq \sum_{k=1}^{k_{\max}} k_{\max}  \leq  k_{\max}^2  .
\end{eqnarray}
From (\ref{eqn:expNumMeas}) we have 
\begin{eqnarray} \nonumber
\mathbb{E}[N] &=&  \frac{\pi E_1 + (1-\pi) E_0 }{ \alpha (1- \pi) +\pi(1  - \beta) }  \\ \nonumber 
&\leq & \frac{\pi k_{\max}^2  + 4(1-\pi) }{ \alpha (1-\pi) + \pi (1-\beta)}\\  \nonumber
& \leq & \frac{ \pi k_{\max}^2 }{(1-\pi) (1-\beta)}  + \frac{ 4 (1-\pi)}{\pi (1-\beta)}
\end{eqnarray}
and finally 
\begin{eqnarray} \nonumber
\lim_{\pi \rightarrow 0} \pi \E[N] \leq C.
\end{eqnarray}

\section*{Appendix D}

\emph{Proof of Theorem \ref{thm:NSLB}:}
Assume that $ P_e \leq \frac{\delta}{1+\delta}$ and from (\ref{eqn:Pe}) we have
\begin{eqnarray} \label{eqn:NSalphimp}
\frac{ \alpha(1-\pi)}{\pi (1-\beta)}  \leq \delta .
\end{eqnarray}
From (\ref{eqn:ENfixed}),
\begin{eqnarray} \label{eqn:NOcond} \nonumber
\mathbb{E}[N]  &\geq &\frac{N_0} {\pi (1+\delta) (1-\beta)   }.
\end{eqnarray} 
Next, for any binary hypothesis test with false negative $\alpha$ and false positive $\beta$, the following identity holds:
\begin{eqnarray} \label{eqn:N0ref}
N_0 \geq \frac{\beta \log \left( \frac{\beta}{1-\alpha} \right)   + (1-\beta) \log\left( \frac{ 1- \beta}{\alpha}\right)  } { D(P_1||P_0)}.
\end{eqnarray}
To see (\ref{eqn:N0ref}), recall that for non-adaptive procedures, $N_0 = E_0 = E_1$, and thus both bounds in  (\ref{eqn:E1bnd}) and (\ref{eqn:E0bnd}) apply.  
This gives
\begin{eqnarray} \nonumber
\E[N] &\geq& \frac{\beta \log \left( \frac{\beta}{1-\alpha} \right)}{   \pi (1+\delta)(1-\beta) D(P_1||P_0)} + \frac{ \log\left( \frac{ 1- \beta}{\alpha}\right)  } { \pi (1+\delta) D(P_1||P_0)} \\ \nonumber
&\geq& \frac{ \log\left(\frac{1-\pi}{\delta \pi } \right) -1 }{ \pi (1+\delta) D(P_1||P_0)} \\ \nonumber
&\geq &  \frac{ \log\left(\frac{1}{2 \delta \pi } \right) -1 }{ \pi (1+\delta) D(P_1||P_0)}  
\end{eqnarray} 
where the second inequality follows from (\ref{eqn:T1bnd7}) and (\ref{eqn:NSalphimp}), and the last inequality as $\pi \leq 1/2$.

\bibliographystyle{IEEEtran}
\bibliography{bestSeq.bib}

\begin{IEEEbiographynophoto}{Matthew L. Malloy}(M'05)received the B.S. degree in Electrical and Computer Engineering from the University of Wisconsin in 2004, the M.S. degree from Stanford University in 2005 in electrical engineering, and the Ph.D. degree from the University of Wisconsin in December 2012 in electrical and computer engineering.   

Dr. Malloy currently holds a postdoctoral research position at the University of Wisconsin in the Wisconsin Institutes for Discovery.  
From 2005-2008, he was a radio frequency design engineering for Motorola in Arlington Heights, IL, USA.  In 2008 he received the Wisconsin Distinguished Graduate fellowship. In 2009, 2010 and 2011 he received the Innovative Signal Analysis fellowship.  

Dr. Malloy has served as a reviewer for the IEEE Transactions on Signal Processing, the IEEE Transactions on Information Theory, IEEE Transactions on Automatic Control and the Annals of Statistics.  He was awarded the best student paper award at the Asilomar Conference on Signals and Systems in 2011.   His current research interests include signal processing, estimation and detection, information theory, statistics and optimization, with applications in biology and communications.
\end{IEEEbiographynophoto}

\begin{IEEEbiographynophoto}{Gongguo Tang}(M,'09) received the B.Sc. degree in mathematics from the Shandong University, China, in 2003, the M.Sc. degree in systems science from the Chinese Academy of Sciences in 2006, and the Ph.D. degree in electrical and systems engineering from Washington University in St. Louis in 2011.

He is currently a Postdoctoral Research Associate at the Department of Electrical and Computer Engineering, University of Wisconsin-Madison. His research interests are in the area of signal processing, convex optimization, information theory and statistics, and their applications.
\end{IEEEbiographynophoto}

\begin{IEEEbiographynophoto}{Robert D. Nowak}(F,'10) received the B.S., M.S., and Ph.D. degrees in electrical engineering from the University of Wisconsin-Madison in 1990, 1992, and 1995, respectively.  He was a Postdoctoral Fellow at Rice University in 1995-1996, an Assistant Professor at Michigan State University from 1996-1999,  held Assistant and Associate Professor positions at Rice University from 1999-2003, and is now the McFarland-Bascom Professor of Engineering at the University of Wisconsin-Madison. 

Professor Nowak has held visiting positions at INRIA, Sophia-Antipolis (2001), and Trinity College, Cambridge (2010). He has served as an Associate Editor for the IEEE Transactions on Image Processing and the ACM Transactions on Sensor Networks, and as the Secretary of the SIAM Activity Group on Imaging Science. He was General Chair for the 2007 IEEE Statistical Signal Processing workshop and Technical Program Chair for the 2003 IEEE Statistical Signal Processing Workshop and the 2004 IEEE/ACM International Symposium on Information Processing in Sensor Networks. 

Professor Nowak received the General Electric Genius of Invention Award (1993), the National Science Foundation CAREER Award (1997), the Army Research Office Young Investigator Program  Award (1999), the Office of Naval Research Young Investigator Program Award (2000), the IEEE Signal Processing Society Young Author Best Paper Award (2000), the IEEE Signal Processing Society Best Paper Award (2011), and the ASPRS Talbert Abrams Paper Award (2012). He is a Fellow of the Institute of Electrical and Electronics Engineers (IEEE). His research interests include signal processing, machine learning, imaging and network science, and applications in communications, bioimaging, and systems biology.
\end{IEEEbiographynophoto}

\end{document}